\documentclass[a4paper,11pt]{article}
\usepackage[utf8]{inputenc}
\usepackage{times,calc,graphicx,natbib,xspace} 
\usepackage{amssymb,amsmath,amsthm,amscd,cite,setspace}
\usepackage[margin=1in]{geometry}
\usepackage{complexity}
\usepackage{algpseudocode,algorithm}
\usepackage{url}
\usepackage{hyperref}
\hypersetup{colorlinks = true,citecolor=blue,urlcolor=purple,hypertexnames=false}

\newtheorem{thm}{Theorem}
\newtheorem{corollary}{Corollary}
\newtheorem{lemma}{Lemma}

\newcommand{\perm}{\mathop{\textup{Per}}}
\newcommand{\EX}{\mathop{{}\mathbb{E}}}
\newcommand{\bigo}[1]{\mathcal{O}(#1)}
\newcommand{\biggo}[1]{\mathcal{O}\bigg(#1\bigg)}

\newcommand{\multichoose}[2]{(^{#1+#2-1}_{\hphantom{#1+}#2})}
\newcommand*\Let[2]{\State #1 $\gets$ #2}
\newcommand{\Bkl}{B^{\diamond}_{k,\ell}}
\newcommand{\Bk}{B^{\diamond}_{k}}

\def\BS/{Boson Sampling}
\title{Faster classical \BS/ }
\author{Peter Clifford\\
 {\normalsize Department of Statistics}\\
  {\normalsize University of Oxford}\\
  {\normalsize United Kingdom}\\
  \and
  Rapha\"el Clifford\\
  {\normalsize Department of Computer Science}\\
  {\normalsize University of Bristol}\\
  {\normalsize  United Kingdom}\\
  }

\date{\vspace{-1ex}}  

\begin{document}

\maketitle
\thispagestyle{empty}

\begin{abstract}
	 Since its introduction \BS/ has been the subject of intense study in the world of quantum computing. The task is to sample independently from the set of all $n \times n$ submatrices built from possibly repeated rows of a larger $m \times n$ complex matrix according to a probability distribution related to the permanents of the submatrices.  Experimental systems exploiting quantum photonic effects  can in principle perform the task at great speed.  In the framework of classical computing, Aaronson and Arkhipov~(2011) showed that exact \BS/ problem cannot be solved in polynomial time  unless the polynomial hierarchy collapses to the third level. Indeed for a number of years the fastest known exact classical algorithm ran in $\bigo{\multichoose{m}{n}\, n 2^n }$ time per sample, emphasising the potential speed advantage of quantum computation. The advantage was reduced by Clifford and Clifford~(2018) who gave a significantly faster classical solution taking $\bigo{n 2^n + \poly(m,n)}$ time and linear space, matching the complexity of computing the permanent of a single matrix when $m$ is polynomial in $n$.
	
   We continue by presenting an algorithm for \BS/ whose average-case time complexity is much  faster when $m$ is proportional to 
$n$. In particular, when $m = n$ our algorithm runs in approximately  $O(n\cdot1.69^n)$ time on average.  This result further increases the problem size needed to establish quantum computational supremacy via \BS/.
\end{abstract}

\setcounter{page}{1}
\section{Introduction}

The search for so-called quantum computational supremacy has garnered a great deal of interest and investment in recent years.  One of the most promising candidates for this goal was introduced at STOC~'11 by \citeauthor{aaronson2011computational} who described an experimental set-up in linear optics known as \BS/.  

Since its introduction, \BS/ has attracted a great deal of attention with numerous experimental efforts around the world attempting implementations for various problem sizes~\citep[see e.g.\@][]{spring2013boson,bentivegna2015experimental,broome2013photonic,tillmann2013experimental,crespi2013integrated,spagnolo2014experimental,LSS:2016}. The ultimate goal is to exhibit a physical quantum experiment of such a scale that it would be hard if not impossible to simulate the output classically and thereby to establish so-called `quantum supremacy'.    In terms of the physical \BS/ experiment, $n$ corresponds to the number of photons and $m$ the number of output modes and increasing either of these is difficult in practice. Progress has therefore been slow \citep[see][and the references therein]{LBR:2017} with the experimental record until recently being $n = 5,  m = 9$ \citep[]{Wang2016}.   However in a breakthrough result in 2019, \citeauthor{Wang:2019} demonstrated a \BS/ experiment with $n=20$ photons and $m=60$. Their \BS/ experimental setup has $20$ input photons of which $14$ are detected, following the model of \citet{AB:2016}.

Phrased in purely mathematical terms, the task is to generate independent random samples from a particular probability distribution on all multisets of size $n$ with elements chosen from $[m]$ as follows. It is convenient to represent a multiset by an array  $\mathbf{z} = (z_1,\dots,z_n)$ consisting of elements of the multiset in non-decreasing order. We denote the set of distinct values of $\mathbf{z}$ by $\Phi_{m,n}$, with $\mu(\mathbf{z}) = \prod_{j=1}^m s_j!$ where $s_j$ is the multiplicity of the value $j$ in $\mathbf{z}$. The cardinality of $\Phi_{m,n}$ is known to be $\multichoose{m}{n}$ -- see \citet{Feller:1968}, for example.

Now let $A^{[n]} = (a_{ij})$ be the complex valued $m \times n$ matrix consisting of the first $n$ columns of a given 
$m$-dimensional  Haar random unitary matrix, $A$. For each $\mathbf{z}$, build an $n \times n$ matrix $A^{[n]}_{\mathbf{z}}$ where the $k$-th row  of $A^{[n]}_{\mathbf{z}}$ is row $z_k$ in $A^{[n]}$ for $k=1,\dots,n$ and define a probability mass function (pmf) on $\Phi_{m,n}$ as

\begin{equation}\label{eq:Bz_s}
q(\mathbf{z}|A) 
= \frac{1}{\mu(\mathbf{z})} \left| \perm A^{[n]}_{\mathbf{z}}\right|^2 
\overset{\text{defn}}{=} \frac{1}{\mu(\mathbf{z})} \left| \sum_{\sigma} \prod_{k=1}^n a_{z_k,\sigma_k}\right|^2, \quad \mathbf{z} \in \Phi_{m,n},
\end{equation}
where $\perm A^{[n]}_{\mathbf{z}}$ is the permanent of $A^{[n]}_{\mathbf{z}}$ and in the definition the summation is for all  $\sigma \in \pi[n]$,  the set of permutations of $[n]$.  For given $A$, the  \textit{\BS/ problem} is to simulate random samples from the pmf $q(\mathbf{z}|A)$ either with a quantum photonic device or with classical computing hardware. 

\citeauthor{aaronson2011computational} show that exact \BS/ is not efficiently solvable by a classical computer unless $\P^{\#\P}  = \BPP^{\NP}$ and the polynomial hierarchy collapses to the third level. Although the original proof restricted the range of $m$ and $n$ a  brute force evaluation of the probabilities of each multi\-set, as a preliminary to random sampling, requires the calculation of $\multichoose{m}{n}$  permanents of $n \times n$ matrices, each one of which takes $\bigo{n 2^n}$ time with the fastest known algorithms. Previously we gave a faster classical sampling algorithm running in  $\bigo{n 2^n + \poly(m,n)}$ time and linear space \citep{clifford2018classical}.    This suggested that at least $50$ photons would be needed in any Boson Sampling experiment to demonstrate so-called quantum computational supremacy.  A more recent fine grained complexity theoretic analysis of  \BS/ suggests that in fact $90$ input photons will be needed to achieve quantum computational supremacy~\citep{DHKLR:2020}.

In practical terms increasing the number of input photons is not the only difficultly when performing \BS/.  Experiments also get increasingly difficult to perform as the number of output modes increases. Current experiments, for example, have typically set the number of modes to be a small multiple of the number of photons and it is likely this trend will continue in the near future.

The original hardness result for exact \BS/ of \citeauthor{aaronson2011computational} required that $m \geq 2n$ as a crucial step in their reduction. For approximate \BS/ the number of outputs has to increase at least quadratically with the number of input photons for any known hardness results to apply. On the other hand, there is no existing evidence that approximate \BS/ is easier than exact \BS/ for any set of parameters.  Moreover, 
 \citet{GS:2018} showed that computing the permanent of a unitary matrix itself is $\#\P$-hard. By applying this result to the proof technique of  \citeauthor{aaronson2011computational} it follows that even with the number of output modes equal to the number of input photons, \BS/ cannot be performed in polynomial time classically unless $\P^{\#\P}  = \BPP^{\NP}$.  This raises the question of whether experimentalists should now focus on increasing the number of input photons alone, keeping the number of output modes to be a small multiple of, or even equal to, the number of input modes.

We answer this question in the negative and show that if $m = \theta n$ for constant $\theta \geq 1$ then it is indeed possible to sample from the \BS/ distribution significantly more quickly than was known before.   We establish the expected complexity averaged over all realisations. For simplicity, we first give the result for the case $m = n$. The time complexity for the general case is given in Section \ref{sec:av-case}.
\vspace{0.5em}
\begin{thm}\label{thm1} The expected time complexity of \BS/ with $m = n$ is
$$\bigo{n \rho ^n},\quad \text{where $\rho = \textstyle{\frac{27}{16}} \approx 1.69$}.$$
The additional space complexity on top of that needed to store the input is $\bigo{m}$.
\end{thm}

Our sampling algorithm also applies directly to the closely related problem of scattershot \BS/ \citep[]{bentivegna2015experimental}. From a mathematical perspective this produces no additional complications beyond the specification of a different $m \times n$ matrix. Once the matrix is specified we can apply our new sampling algorithm directly. 

\section{Related work}
\citet{TD:2004} were the first to recognise that studying the complexity of sampling from low-depth quantum circuits could be useful for demonstrating a separation between quantum and classical computation. Since that time the goal of finding complexity separations for quantum sampling problems has received a great deal of interest. We refer the interested reader to \citet{LBR:2017} and \citet{HM:2017} for surveys on the topic.

The search for faster classical algorithms for \BS/ has been of central interest since the problem was first explicitly formulated  by \citeauthor{aaronson2011computational}.  The first significant breakthrough was a Markov chain Monte Carlo (MCMC) sampling procedure developed by \citet{neville2017classical}. The overall approach of MCMC is to take samples from some easy to compute proposal distribution and then to accept them depending on their individual likelihoods.  \citeauthor{neville2017classical} were able to provide numerical evidence that for limited problem sizes only approximately $200$ such permanent computations were needed to take one sample approximately from the \BS/ distribution. This MCMC approach is however necessarily approximate and does not give provable bounds on the quality of its approximation.

In a previous paper \citep{clifford2018classical} we presented a faster and provably correct \BS/ algorithm running in  $\bigo{n 2^n + \poly(m,n)}$ time per sample,  costing approximately two permanent calculations of $n \times n$ matrices for all values of $n$. In recent papers the algorithm has been extended and adapted to examine the effect of non-uniform losses and binned input/output modes. See \citet{moylett2019classically,Shchesnovich:2019,brod2020classical} and the references therein.

\section{Methods}

\subsection{Computing the permanent of low rank matrices}

The complexity of computing the permanent of an arbitrary $k \times k$ complex matrix was shown to be $\bigo{k^2 2^k}$~\citep{Ryser:1963} and subsequently $\bigo{k 2^k}$ \citep{NW:1978,Glynn:2010}. The computation time is decreased when there are several identical rows or columns, resulting in a matrix of reduced rank~\citep{,Tichy:2011,shchesnovich2013asymptotic,Shchesnovich:2019}.  We will show how these ideas can be implemented in practice to produce a faster algorithm that suits our needs more closely. 

Let $A^{[k]}$ be the first $k$ columns of a complex $m \times m$ matrix, $A$. Assume the rows of $A^{[k]}$ are distinct and let $B$ be a $k \times k$  matrix with repeated rows drawn from $A^{[k]}$. Specifically, let $\mathbf{z}$ be a multiset of size $k$ with elements from $[m]$ and let $\mathbf{s} = (s_1,\dots,s_m)$ be its associated array of multiplicities, then the $i$th row of $B$ will be row $z_i$ of $A^{[k]}$, for $i = 1,\dots,k$. 

 According to Ryser's formula
\begin{equation}\label{eq:ryser1}
\perm B = (-1)^k \sum_{T \subseteq [k]}(-1)^{|T|} \prod_{j=1}^k \sum_{i \in T} b_{i,j}.
\end{equation}
The subset $T$ can be written as $\cup_{\nu=1}^m T_\nu$ where $T_\nu = \{i:z_i = \nu\}$.  Note there are $s_\nu$ elements in $\{i:z_i = \nu\}.$  It follows there are $\binom{s_\nu}{r_\nu}$ ways of choosing a subset $T_\nu$ of a given size $r_\nu$, when $r_\nu$ is between $0$ and $s_\nu$. Furthermore $b_{i,j} = a_{\nu,j}$ whenever $i \in T_\nu$. 

Ryser's formula can then be expressed as
\begin{equation}\label{eq:ryser2}
\perm B = (-1)^k \sum_{r_1 = 0}^{s_1} \cdots \sum_{r_m = 0}^{s_m} (-1)^{r_1 + \cdots +r_m}\prod_{\nu=1}^m \binom{s_\nu}{r_\nu}  \prod_{j=1}^k \left(\sum_{\nu=1}^m r_\nu a_{\nu,j}\right).
\end{equation}

A straightforward implementation of the formula requires us to iterate over the tuples $(r_1,\dots, r_m)$. Since each term in the summation requires $\bigo{k m}$ operations, the run time is then $\bigo{k m \prod_{\nu=1}^m (s_\nu+1)}$ as shown by \citet{shchesnovich2013asymptotic}. Similar complexity is achieved but with an improved constant factor overhead when starting from Glynn's formula \citep{chin2018generalized}. 

To make this iteration as fast as possible we would ideally like to perform the iteration in such a way that at each stage we change only one of the $r_i$ values and these are changed by $\pm 1$.  This can be achieved using Guan codes (otherwise known as generalised Gray codes)~\citep{Guan:1998}, borrowing from the idea of \citet{NW:1978} who used a basic Gray code to speed up Ryser's algorithm.
\vspace{0.5em}

\begin{thm}\label{thm:guan} Using Guan codes, $\perm B$ can be calculated in $\bigo{k \prod_{\nu=1}^m (s_\nu+1)}$ time.
\end{thm}
\begin{proof}
	There are $\prod_{\nu=1}^m (s_\nu+1)$ terms in the outside set of summations in \eqref{eq:ryser2}.  By using Guan's algorithm we can move through the tuples $(r_1,\dots,r_m)$ exhaustively, adding or subtracting $1$ from a single element. This means that the set of values $\sum_{\nu=1}^m r_\nu a_{\nu,j}, j \in [k]$ can be updated in $k$ operations. The product over $j \in [k]$ is a further $k$ operations. Updating the relevant binomial term is a $\bigo{1}$ operation. The total operation count is then $\bigo{k \prod_{\nu=1}^m (s_\nu+1)}$.
\end{proof}

\subsection{Laplace expansion}\label{sec:Laplace}
We make extensive use of the Laplace expansion for permanents \citep[see][page 578]{marcus1965permanents}, namely that for any $k \times k$ matrix $B = (b_{i,j})$,
\begin{equation}\label{eq:Laplace}
\perm B = \sum_{\ell = 1}^k b_{k,\ell} \perm \Bkl, \nonumber
\end{equation}
where $\Bkl$ is the submatrix with row $k$ and column $\ell$ removed.  Note that $\Bkl$ only depends on $\Bk$ the submatrix of $B$ with the $k$-th row removed.  An important consequence is that when $B$ is modified by changing its  $k$-th row, the modified permanent can be calculated in $\bigo{k}$ steps,  provided the values $\{\perm \Bkl\}$ are available. As explained in \citet{clifford2018classical}, we can take advantage of this observation to quickly compute a set of permanents of matrices each with one row differing from the other.

We now show that computation of all of the values $\{\perm \Bkl, \ell \in [k]\}$ has the same time complexity as computing $\perm B$, the permanent of a single $k \times k$ matrix when there are repeated rows in $B$.   We derive this result for matrices with repeated rows using a combination of Ryser's algorithm and Guan codes.
   
\vspace{0.5em}
\begin{lemma}\label{minors}
	Let $B$ be a $k \times k$ complex matrix with repeated rows specified by multiplicities $\mathbf{s}$ and let $\{\Bkl\}$ be submatrices of $B$ with row $k$ and column $\ell$ removed, $\ell \in [k]$. The collection $\{\perm  \Bkl, \ell \in [k]\}$ can be evaluated jointly with the same time complexity as that of $\perm B$, with $\bigo{k}$ additional space.  
\end{lemma}
\begin{proof}[Proof of Lemma]
	By applying Ryser's formula ~\eqref{eq:ryser2} to $\Bkl$ for a given value of $\ell$ we have:
	
	\begin{equation}
	\perm  \Bkl = (-1)^k \sum_{r_1 = 0}^{s_1^\diamond} \cdots \sum_{r_m = 0}^{s_m^\diamond} (-1)^{r_1 + \cdots +r_m}\prod_{\nu=1}^m \binom{s_\nu^\diamond}{r_\nu}  \prod_{j \in [k] \setminus \ell}  w_j(\mathbf{r}).
	\end{equation}
	where $\mathbf{s}^\diamond$ is the multiplicity of repeated rows in $\Bk$ and $w_j(\mathbf{r}) = \sum_{\nu=1}^m r_\nu a_{\nu,j}$. 
	
	Working through values of $\mathbf{r}$ in Guan code order, the terms $\{w_j(\mathbf{r}), j \in [k]\}$ can be evaluated in $\bigo{k}$ combined time for every new $\mathbf{r}$. This is because successive $\mathbf{r}$ arrays differ by one in a single element. The product of the $w_j(\mathbf{r})$ terms can be computed in $\bigo{k}$ time giving $\bigo{k\prod_{\nu=1}^m (s^\diamond_\nu+1)}$ time to compute $\perm  \Bkl$ for a single value of $\ell$, but of course this has to be replicated $k$ times to cover all values of $\ell$. 
	
	To compute $\{\perm  \Bkl,  \ell \in [k]\}$ more efficiently we observe that each product $\prod_{j \in [k] \setminus \ell}  w_j(\mathbf{r})$ can be expressed as $f_\ell b_\ell$  where $f_\ell = \prod_{j=1}^{\ell-1} w_j(\mathbf{r}), \ell = 2,\dots,k$ and $b_\ell = \prod_{j=\ell+1}^k  w_j(\mathbf{r}), \ell = 1,\dots,k-1$ are forward and backward cumulative products, with $f_1 = b_k = 1$.  
	
	We can therefore compute all of the partial products $\prod_{j \in [k] \setminus \ell}   w_j(\mathbf{r})$ in $\bigo{k}$ time, giving an overall total time complexity of $\bigo{k\prod_{\nu=1}^m (s^\diamond_\nu+1)}$ for jointly computing $\{\perm  \Bkl, \ell \in [k]\}$, and since $\mathbf{s}^\diamond \leqslant \mathbf{s}$, the time complexity is as claimed. Furthermore the computation time has constant factor overheads similar to that of computing $\perm B$. Other than the original matrix, space used is dominated by the two arrays of cumulative products, both of length $\bigo{k}$. 
\end{proof}  


\section{\BS/ algorithm}

\citet{clifford2018classical} provide the following algorithm for \BS/:
\begin{algorithm}
	\caption{Boson sampler: single sample $\mathbf{z}$ from $q(\mathbf{z}|A)$ 
		\label{alg:better}}
	\begin{algorithmic}[1]
		\Require{$m$ and $n$ positive integers;  $m$-dimensional Haar random unitary matrix, $A$}
		\State $\mathbf{r} \gets \varnothing$ \Comment{\sc Empty array}
		\State $A^{[n]} \gets \textrm{Permute}(A^{[n]})$ \Comment{Randomly permute columns of $A^{[n]}$}
		\State $w_i \gets |a_{i,1}|^2, i \in [m]$ \Comment{Make indexed weight array $w$}
		\State $x \gets \textrm{Sample}(w)$ \Comment{Sample index $x$ from $w$}
		\State $\mathbf{r} \gets (\mathbf{r}, x)$ \Comment{Append $x$ to $\mathbf{r}$}
		\For{$k \gets 2 \textrm{ to } n$}
		\State $\Bk \gets A^{[k]}_{\mathbf{r}}$
		\State \sc Compute $\{\perm \Bkl, \ell \in [k] $\} \Comment{As Lemma \ref{minors}}
		\State $w_i \gets |\sum_{\ell = 1}^k a_{i,\ell} \perm \Bkl|^2,\; i \in [m]$ \Comment{Using Laplace expansion}
		\State $x \gets \textrm{Sample}(w)$ 
		\State $\mathbf{r} \gets (\mathbf{r}, x)$
		\EndFor
		\Let{$\mathbf{z}$}{\sc IncSort($\mathbf{r}$)} \Comment{Sort $\mathbf{r}$ in non-decreasing order}
		\State \Return{$\mathbf{z}$}
	\end{algorithmic}
\end{algorithm}

Calculation of the time complexity proceeds as in \citet{clifford2018classical}, with a few modifications to take account of repeated rows in evaluating permanents of the submatrices involved.  At stage $k$, let $\mathbf{s}^{(k)}$ be the multiplicities of the values in $(r_1,\dots,r_k)$.  The operation count in applying Lemma \ref{minors} is $\bigo{k \prod_{\nu=1}^m (s_\nu^{(k)}+1)}$. Using the Laplace expansion for permanents, as described in Section \ref{sec:Laplace}, an array of length $m$ is obtained by summing $k$ terms for each $r_k \in [m]$.  Taking a single sample from the pmf proportional to the array takes $\bigo{m}$ time. This gives the time complexity bound for stage $k$ of $\bigo{k \prod_{\nu=1}^m (s_\nu^{(k)}+1)}$ + $\bigo{mk}$ and hence a total operation count of 
\begin{equation}\label{eq:complexity}
\bigo{\textstyle{\sum_{k=1}^n k \prod_{\nu=1}^m (s_\nu^{(k)}+1)}} + \bigo{mn^2}.
\end{equation}

Importantly at the conclusion of stage $k$, the $k \times k$ matrix $A_{r_1,\dots, r_k}^{[k]}$ has been found.  In other words at this intermediate stage a random submatrix has been drawn for the \BS/ problem with $k$ photons and $m$ output modes. Equivalently, the multiset formed by sorting $(r_1,\dots, r_k)$ in non-decreasing order is then a random sample from the \BS/ distribution on $\Phi_{m,k}$. 

We now consider the average-case time complexity when the algorithm is applied with a random choice of Haar unitary matrix, $A$.

\section{Marginal uniformity of the \BS/ distribution}
Arkhipov and Kuperberg have shown that the marginal \BS/ pmf averaged over Haar random unitaries is uniform on the space of multisets.     
\vspace{0.5em} 
\begin{thm}\label{thm:AK}
With $q(\mathbf{z}|A)$ as in \eqref{eq:Bz_s} and $A$ a random matrix drawn from the $m$-dimensional Haar random unitary distribution, the marginal \BS/ pmf $\mathcal{Q}(\mathbf{z})$ is given by
\begin{equation}\label{eq:AK}
\mathcal{Q}(\mathbf{z}) = \EX q(\mathbf{z}|A) = \binom{m+n-1}{n}^{\!-1}\!, \quad \mathbf{z} \in \Phi_{m,n}.
\end{equation}
\end{thm}

The proof is immediate from quantum theoretic considerations \citep{AG:2012}.  It can also be derived from properties of the Weingarten function in random matrix theory as follows.

Suppose $A$ is an $m$-dimensional Haar random unitary matrix. The Weingarten function is defined to be
$$ W\!g(\alpha,m) = \EX (A_{1,1}\cdots A_{1,n}\bar{A}_{1,\alpha(1)}\cdots \bar{A}_{n,\alpha(n)}),$$
for $n \leqslant m$ and $\alpha$ in $\pi[n]$, where $\bar{A}_{i,j}$ is the complex conjugate of $A_{i,j}$ and $\pi[n]$ is the set of permutations of $[n]$. 

We make use of the following result
\vspace{1em}
\citep{collins2003moments,collins2006integration}.
\begin{thm}\label{thm:Collins}
	
	Let $\mathbf{i} = (i_1,\dots,i_n)$, $\mathbf{i'} = (i'_1,\dots,i'_n)$, $\mathbf{j} = (j_1,\dots,j_n)$ and $\mathbf{j'} = (j'_1,\dots,j'_n)$ be arrays of positive integers  and let
	$$ \delta_\alpha(\mathbf{i,i'}) = \prod_{k=1}^n \delta(i_k,i_{\alpha(k)}) \quad \text{for $\alpha \in \pi[n]$}.$$  
	then
	$$\EX (A_{i_1,j_1}\cdots A_{i_n,j_n}\bar{A}_{i'_1,j'_1}\cdots \bar{A}_{j'_n,j'_n}) 
	= \sum_{\alpha,\beta \in \pi[n]} \delta_\alpha(\mathbf{i,i'}) \delta_\beta(\mathbf{j,j'}) W\!g(\alpha \beta^{-1},m)$$
\end{thm}  
As a immediate corollary we have:
\vspace{0.5em} 
\begin{corollary}\label{cor:X}
	$$\sum_{\alpha \in \pi[n]} W\!g(\alpha,m) = \frac{1}{m(m+1)\cdots(m + n -1)}.$$
\end{corollary}
\begin{proof}
	Apply Theorem~\ref{thm:Collins} to $\EX(|A_{1,1}|^{2n})$ with $\mathbf{i}=\mathbf{i'} = \mathbf{j}=\mathbf{j'}= (1,\dots,1)$, so that 
	$$\EX( |A_{1,1}|^{2n}) = \EX(A_{1,1}\cdots A_{1,1}\bar{A}_{1,1}\cdots \bar{A}_{1,1}) = \sum_{\alpha,\beta \in \pi[n]} W\!g(\alpha \beta^{-1},m) = n!\sum_{\alpha \in \pi[n]} W\!g(\alpha ,m).$$
	The last reduction follows since  $\{\alpha \beta^{-1} : \alpha \in \pi[n]\} = \pi[n]$ for any given $\beta \in \pi[n]$. 
	
	From ~\citep{PetzRaffey:2004}, for example, with $W = |A_{1,1}|^2$  we have
	$$\EX (|A_{1,1}|^{2n}) = \EX (W^n) = \int_0^1 (m-1) w^n (1-w)^{m-2} dw = \frac{n!}{m(m+1)\cdots (m+n-1)},$$
	and the proof is complete.   
\end{proof}
Turning now to the \BS/ distribution on multisets of size $n$ with elements in $[m]$ where a multiset is represented by an array  $\mathbf{z} = (z_1,\dots,z_n)$ consisting of elements of $[m]$ in non-decreasing order.  
As before $\mu(\mathbf{z}) = \prod_{j=1}^m s_j!$ where $s_j$ is the multiplicity of the value $j$ in $\mathbf{z}$.

From the definition \eqref{eq:Bz_s} we now have
\begin{align}
	\mu(\mathbf{z}) \mathcal{Q}(\mathbf{z})
	&= \EX \left| \sum_{\sigma} \prod_{k=1}^n A_{z_k,\sigma_k}\right|^2, \sigma \in \pi[n] \nonumber\\
	&= \EX \left(\sum_\sigma \prod_{k=1}^n A_{z_k,\sigma_k }\right)\left(\sum_\tau \prod_{k=1}^n \bar{A}_{z_k,\tau_k}\right),\tau \in \pi[n] \nonumber\\
	&= \sum_{\sigma,\tau \in \pi[n]} \EX \left( \prod_{k=1}^n A_{z_k,\sigma_k} \bar{A}_{z_k,\tau_k}. \right)\nonumber\\ 
	&= \sum_{\sigma,\tau \in \pi[n]} \sum_{\alpha,\beta \in \pi[n]} 
	\delta_\alpha(\mathbf{z, z}) \delta_\beta(\sigma, \tau) W\!g(\alpha \beta^{-1},m)\nonumber \\
	\label{eq:final}
	&= \sum_{\alpha,\beta \in \pi[n]} \delta_\alpha(\mathbf{z, z}) W\!g(\alpha \beta^{-1},m) \sum_{\sigma,\tau \in \pi[n]} \delta_\beta(\sigma, \tau)
	\end{align}
	using Theorem \ref{thm:Collins}. For each $\beta$ the final summation in \eqref{eq:final} is $n!$, so the expectation becomes
	\begin{align} 
	\mu(\mathbf{z}) \mathcal{Q}(\mathbf{z})
	&= n! \sum_{\alpha,\beta \in \pi[n]} \delta_\alpha(\mathbf{z, z}) W\!g(\alpha \beta^{-1},m)\nonumber\\
	&= n! \sum_{\alpha \in \pi[n]} \delta_\alpha(\mathbf{z, z}) \sum_{\beta \in \pi[n]} W\!g(\alpha \beta^{-1},m)\nonumber\\
	&= n! \sum_{\alpha \in \pi[n]} \delta_\alpha(\mathbf{z, z}) \frac{1}{m(m+1)\cdots(m + n -1)}\nonumber\\
	&= \binom{m+n-1}{n}^{-1} \sum_{\alpha \in \pi[n]} \delta_\alpha(\mathbf{z, z})\nonumber
	\end{align}
	from Corollary \ref{cor:X} and collecting factorial terms. The last term counts the number of ways that the elements of the array $\mathbf{z}$ can be permuted without changing the array, so that  
	$\sum_{\alpha \in \pi[n]} \delta_\alpha(\mathbf{z, z}) = \mu(\mathbf{z}),$ and the result \eqref{eq:AK} follows.

\section{Average-case time complexity of \BS/}\label{sec:av-case}

The average-case complexity of Algorithm A is the expected value of \eqref{eq:complexity} when $A$ is drawn from the Haar random unitary distribution. We start by considering the term $\prod_{\nu=1}^m (s_\nu^{(n)} + 1)$. Since the multiplicity array 
$\mathbf{s}^{(n)}$ is an alternative representation of the multiset $\mathbf{z}$, we know from Theorem \ref{thm:AK} that $\mathbf{s}^{(n)}$ is uniformly distributed over the set of all multiplicity arrays, $\Phi_{m,n}^*$. 

Recall that $|\Phi_{m.n}^*| = \binom{m+n-1}{m-1}$ or equivalently $\binom{m+n-1}{n}$ as can be shown with the usual ``stars and bars'' argument of \citet{Feller:1968}, i.e.\ we place $m-1$ bars at locations in $[1,\dots,m+n-1]$ and $n$ stars at the remaining locations.  The total number of such arrangements is $\binom{m+n-1}{m-1}$. Adding further bars at each end, i.e.\ at locations $0$ and $m+n$, the associated array $\mathbf{s}^{(n)}$ is the count of stars between bars.

\vspace{0.5em}
\begin{lemma}\label{lm:multiplicity} Suppose that $\mathbf{s}^{(n)}$ is sampled uniformly from $\Phi_{m,n}^*$, the set of all possible multiplicity arrays then 
\begin{equation}\label{eq:multiplicity} \EX \left(\prod_{\nu=1}^m(s_\nu^{(n)} + 1)\right) =  \binom{ m+ n-1}{n}^{\!-1} \binom{2 m + n -1}{n}.
\end{equation}
\end{lemma}

\begin{proof}
	To see this, start with the stars and bars arrangement for a particular array $\mathbf{s}^{(n)}$ as above. Now consider adding a new bar between each existing neighbouring pair of bars.  If there are $s_i^{(n)}$ stars between a pair of bars, the new bar can be located in $1$ of $s_i^{(n)} + 1$ places, for example if there is one star the new bar can be before or after it.  The number of arrangements of stars and bars for a given array $\mathbf{s}^{(n)}$ is then $\prod_{\nu=1}^m (s^{(n)}_\nu +1)$.
	
Since $\mathbf{s}^{(n)}$ is uniformly distributed on $\Phi_{m,n}^*$ and $|\Phi_{m,n}^*| = \binom{m+n-1}{n}$, 
	$$\EX \left(\prod_{\nu=1}^m(s^{(n)}_\nu + 1)\right) =  \binom{m+n-1}{n}^{-1} \sum_{\mathbf{s}^{(n)} \in \Phi^*_{m,n}} 
	\prod_{\nu=1}^m(s^{(n)}_\nu + 1).$$
	The result now follows because $\sum_{\mathbf{s}^{(n)} \in \Phi_{m,n}^*} \prod_{\nu=1}^m(s^{(n)}_\nu + 1)$ is the total number of arrangements of the new and old stars and bars, i.e.\ the number of ways of placing $2m -1$ bars among $2 m + n -1$ integer locations. 
\end{proof}
\begin{corollary}
	With the conditions of Lemma \ref{lm:multiplicity} and supposing that $m = \theta n$ for some fixed value $\theta$ then 
	$$\EX \left(\prod_{\nu=1}^m(s_\nu^{(n)} + 1)\right) \sim \left(\frac{2(\theta+1)}{2\theta+1}\right)^{1/2} \left[\frac{(2 \theta+1)^{2 \theta+1}}{(4 \theta)^\theta(\theta+1)^{\theta+1}}\right]^n \quad \text{as $ n \to \infty$}.$$
	In particular $\EX \prod_{\nu=1}^m(s_\nu^{(n)} + 1) \sim (4/3)^{1/2} (27/16)^n \approx 1.15 (1.69)^n$ when $\theta = 1$, and $\EX \prod_{\nu=1}^m(s_\nu^{(n)} + 1) \sim  2^n$ as $\theta \to \infty.$ 
\end{corollary}
\begin{proof} Apply Stirling's formula to the factorial terms in \eqref{eq:multiplicity}. \end{proof}
\vspace{0.5em}
\begin{corollary} Let $B$ a $k \times k$ complex matrix with repeated rows having multiplicity $\mathbf{s}^{(k)}$ where $\mathbf{s}^{(k)}$ is uniformly distributed on $\Phi_{m,k}^*$ then the expected running time of $\perm B$ is of order
	$$ k \binom{ m+ k-1}{k}^{\!-1} \binom{2 m + k -1}{k}. $$
\end{corollary}
\begin{proof} This follows directly from Lemma \ref{lm:multiplicity} and Theorem \ref{thm:guan}. \end{proof}
We now prepare to prove the generalisation of Theorem~\ref{thm1}.
\begin{thm}\label{thm:1.1}
	The average-case time complexity of \BS/ is bounded above by a term of order 
	$$ \frac{n(m+n)}{m} \binom{ m+ n}{n+1}^{\!-1} \binom{2 m + n}{n+1} + n^2m \quad \text{as $m,n \to \infty$}. $$
	In particular when $m= \theta n$ for a fixed value of $\theta \geqslant 1$ this is of order 
	$$n \left[\frac{(2 \theta+1)^{2 \theta+1}}{(4 \theta)^\theta(\theta+1)^{\theta+1}}\right]^n \quad \text{as $ n \to \infty$}.$$
\end{thm} 
\begin{proof}
Note that the first term in the time complexity of Algorithm A given in \eqref{eq:complexity}	is equivalent to the total time complexity for evaluating the set of permanents of $A_{r_1,\dots, r_k}^{[k]}$, $k \in [n]$.  In particular, evaluation of the permanent of the final matrix $A_\mathbf{r}^{[n]}$ has time complexity $\bigo{n\prod_{\nu=1}^m(s_\nu^{(n)}+1)}$ where $\mathbf{s}^{(n)}$ is the array of multiplicities in $\mathbf{r}$.  From Lemma \ref{lm:multiplicity} this has a simple reduction in the average case, since $\mathbf{s}^{(n)}$ is uniformly distributed on $\Phi^*_{m,n}$ by Theorem \ref{thm:AK}.  

At the intermediate stage $k$, the matrix $A_{r_1,\dots, r_k}^{[k]}$ can be viewed as the final matrix in a \BS/ algorithm with reduced size, i.e.\ where $k$ columns are taken from $A$ rather than $n$.  Again the average-case complexity has a simple form from Lemma \ref{lm:multiplicity} so the expected value of the first term in the time complexity of Algorithm A given in \eqref{eq:complexity}, averaging over $A$, is 
\begin{align*} 
\EX \sum_{k=1}^n k \prod_{\nu=1}^m (s_\nu^{(k)}+1)  
&= \sum_{k=1}^n  k \binom{ m+ k-1}{k}^{\!-1} \binom{2 m + k -1}{k}\\
&= \frac{[n(m+1)-m+1](m+n)}{(m+1)(m+2)}\binom{m+ n}{n+1}^{\!-1} \binom{2 m + n}{n+1} + \frac{2 m(m-1)}{(m+1)(m+2)}\\
&= \biggo{\frac{n(m+n)}{m} \binom{ m+ n}{n+1}^{\!-1} \binom{2 m + n}{n+1}}. 
\end{align*}
Incorporating the second term in \eqref{eq:complexity} gives the complexity as claimed. Finally, Stirling's formula gives the asymptotic form. 
\end{proof}


\bibliographystyle{apalike}
\bibliography{fasterBS}
\end{document}